\documentclass{article}
\usepackage{ijcai25}
\usepackage{amsfonts}
\usepackage{times}
\usepackage{soul}
\usepackage{url}
\usepackage[hidelinks]{hyperref}
\usepackage[utf8]{inputenc}
\usepackage[small]{caption}
\usepackage{graphicx}
\usepackage{amsmath}
\usepackage{amsthm}
\usepackage{booktabs}
\usepackage{algorithm}
\usepackage{algorithmic}
\usepackage[switch]{lineno}
\usepackage{epstopdf}
\usepackage{wrapfig} 

\newtheorem{theorem}{Theorem}

\title{Se-HiLo: Noise-Resilient Semantic Communication with High-and-Low Frequency Decomposition}

\author{
Zhiyuan Xi$^1$\and
Kun Zhu$^2$\and
Yuanyuan Xu$^{2}$
\emails
xizhiyuan@nuaa.edu.cn,
zhukun@nuaa.edu.cn,
yuanyuan\_xu@hhu.edu.cn,
}

\begin{document}

\maketitle

\begin{abstract}
Semantic communication has emerged as a transformative paradigm in next-generation communication systems, leveraging advanced artificial intelligence (AI) models to extract and transmit semantic representations for efficient information exchange. Nevertheless, the presence of unpredictable semantic noise, such as ambiguity and distortions in transmitted representations, often undermines the reliability of received information. Conventional approaches primarily adopt adversarial training with noise injection to mitigate the adverse effects of noise. However, such methods exhibit limited adaptability to varying noise levels and impose additional computational overhead during model training. To address these challenges, this paper proposes Noise-Resilient \textbf{Se}mantic Communication with \textbf{Hi}gh-and-\textbf{Lo}w Frequency Decomposition (Se-HiLo) for image transmission. The proposed Se-HiLo incorporates a Finite Scalar Quantization (FSQ) based noise-resilient module, which bypasses adversarial training by enforcing encoded representations within predefined spaces to enhance noise resilience. While FSQ improves robustness, it compromise representational diversity. To alleviate this trade-off, we adopt a transformer-based high-and-low frequency decomposition module that decouples image representations into high-and-low frequency components, mapping them into separate FSQ representation spaces to preserve representational diversity. Extensive experiments demonstrate that Se-HiLo achieves superior noise resilience and ensures accurate semantic communication across diverse noise environments.

\end{abstract}

\section{Introduction}

The information theory of Shannon and Weaver~\cite{Shannon,Weaver} laid the foundation for the development of modern communication technologies. In their theory, communication is categorized into three levels: syntax, semantics, and pragmatics. These levels respectively focus on the accurate transmission of symbols, the precise conveyance of the meanings or intentions of the transmitted information, and the effective utilization of that information.

Fueled by rapid advancements in artificial intelligence (AI) technologies, the conventional syntax-focused communication paradigm is progressively evolving into a semantic-driven paradigm~\cite{survey_1,survey_2}. Recent years have witnessed a surge of research efforts dedicated to semantic communication, resulting in the development of various deep neural network-based semantic communication systems~\cite{semantic_sys_1,semantic_sys_2}.

Despite the promising capabilities of neural networks in extracting semantic representations from source data and transmitting them over physical channels, these representations are inherently susceptible to various forms of semantic noise perturbations. Such perturbations stem from natural interference in electronic devices or deliberate malicious jamming, both of which can severely degrade decoding performance at the receiver. As illustrated in Fig.~\ref{semantic_noise}, semantic noise distorts transmitted representations, causing noticeable deviations in reconstructed images and significantly reducing classification accuracy.
\begin{figure}[htbp]
    \centering
    \noindent\includegraphics[width=3.4in]{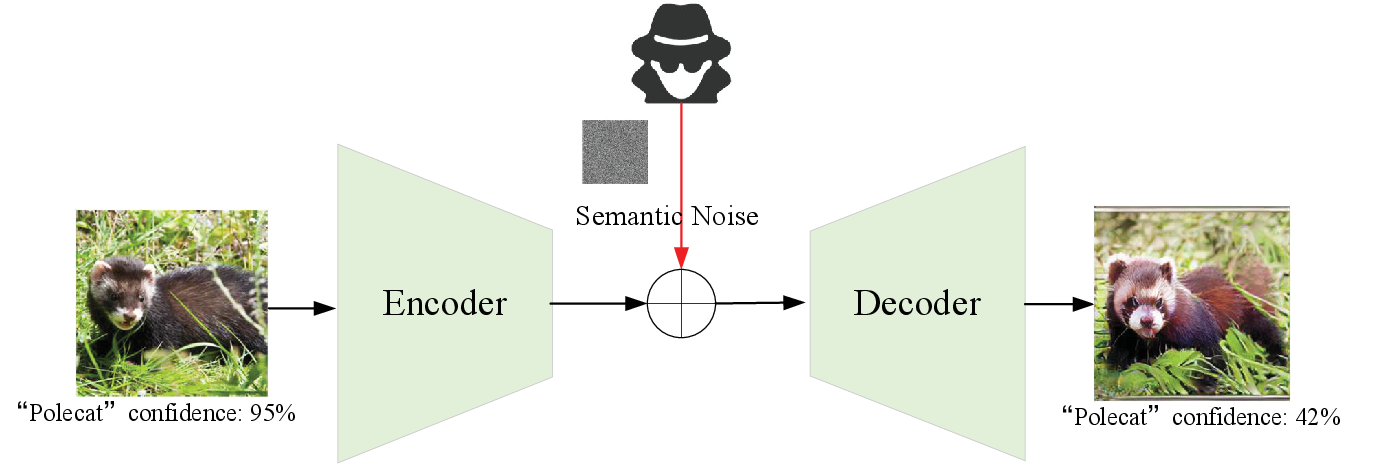}
    \caption{The Impact of Semantic Noise}
    \label{semantic_noise}
\end{figure}

Extensive research~\cite{semantic_noise_1,semantic_noise_2} has been conducted on noise-robust techniques, most of which focusing on adversarial training strategies to introduce noise (e.g., Gaussian white noise) during the training phase. This process encourages neural networks to learn noise-resilient features, thereby enhancing their robustness to semantic noise in communication systems. While conceptually intuitive, these approaches face inherent limitations when applied to real-world scenarios where noise patterns are highly unpredictable and dynamically evolving. Consequently, these methods often struggle to generalize across diverse noise environments,  requiring frequent retraining to accommodate specific noise levels. Such retraining not only incurs substantial computational costs but also undermines scalability, posing significant challenges for the practical deployment of semantic communication systems.


To investigate the impact of semantic noise on semantic representations, we conducted a preliminary experiment using the ResNet34 model~\cite{resnet}. A 10-class image dataset was fed into the model, with Gaussian white noise  progressively added to the output vectors, reducing the signal-to-noise ratio (SNR) from 10 dB to -5 dB. The resulting semantic representations were visualized in 2D space using t-SNE for dimensionality reduction to analyze their distribution. As depicted at the top of Fig.~\ref{tsne}, when noise is absent or minimal, the representations of different categories are well-separated in the representation space, making them easily distinguishable. However, as the noise level increases, the representations from different categories gradually become entangled and intermixed in the representation space, significantly reducing their separability and resulting in a sharp drop in classification accuracy.
\begin{figure*}[htbp]
    \centering
    \noindent\includegraphics[width=5.8in]{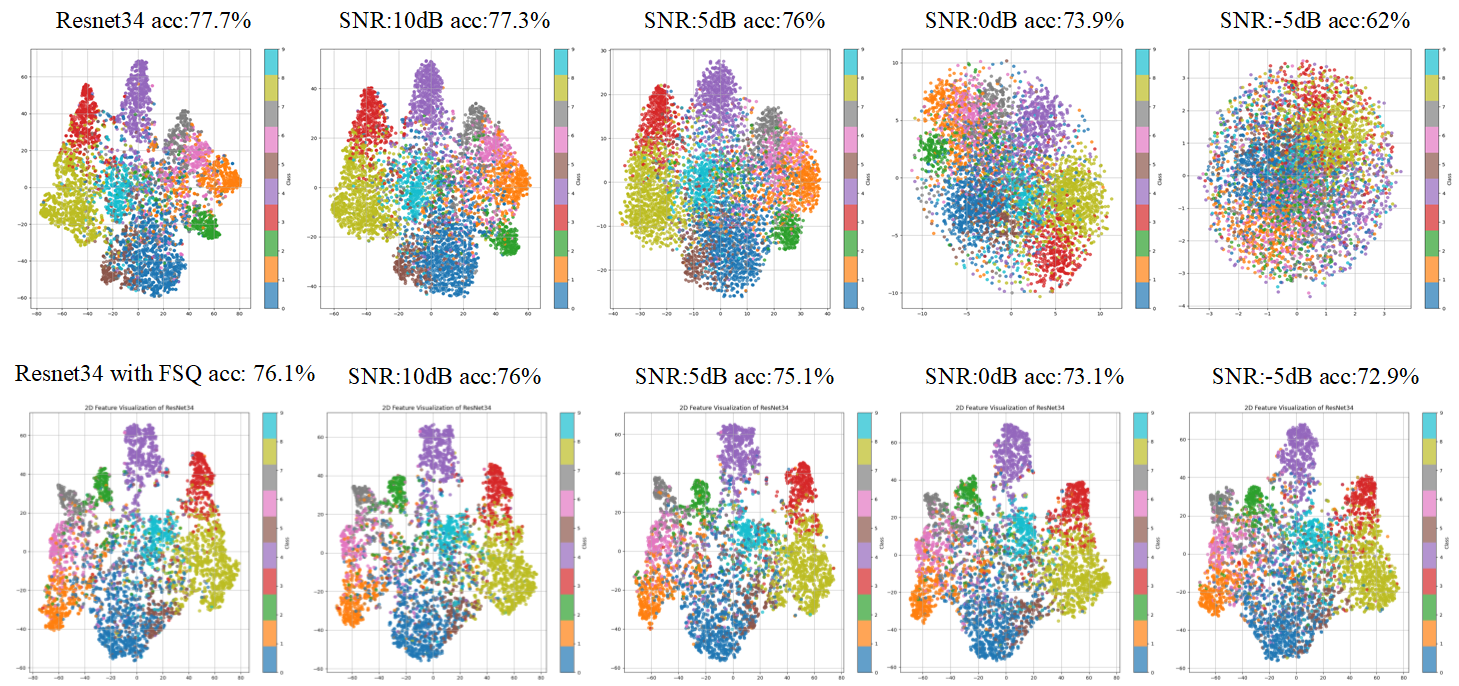}
    \caption{t-SNE Visualization: ResNet34 (Top) vs. ResNet34 with FSQ (Bottom) under Varying SNRs}
    \label{tsne}
\end{figure*}

This observation suggests that noise disrupts the distribution of semantic representations. By constraining these representations to remain within specific regions of the feature space, the impact of noise can be mitigated to some extent. Building on this concept, we design a \textbf{Finite Scalar Quantization (FSQ)}~\cite{fsq} based noise-resilient module, which discretizes each dimension of the feature space into a finite set of values. For instance, consider a 3D vector. With FSQ, each component of the vector is restricted to three possible values, resulting in $3 \times 3 \times 3 = 27$ distinct vector combinations in the space. These combinations, referred to as reference vectors, serve as anchors to which all 3D vectors are mapped via rounding. Semantic representations distorted by noise during transmission, are constrained by FSQ. This method effectively reduces noise impact and enables the system to handle diverse noise levels without requiring artificial noise injection during training.

We integrated the FSQ module into ResNet34 and repeated the same noise injection experiment. As illustrated at the bottom of Fig.~\ref{tsne}, the semantic representations of different categories remain well-separated in the representation space, and the classification accuracy exhibits a more gradual decline as the noise level increases.

However, we observed that the application of FSQ leads to a noticeable reduction in classification performance. This degradation stems from the constraints enforced on the representation space, which limit the model to capture diverse and fine-grained features. In the initial design, uniform constraints were applied across the entire representation, inevitably resulting in information loss. Neural networks inherently exhibit a preference for learning low-frequency components~\cite{dnn_low_frequency}, such as global structures, which are often more redundant. In contrast, high-frequency components, including fine details and textures, are typically less emphasized. Due to the constraints imposed by FSQ, these high-frequency features become more susceptible to degradation.

This imbalance in frequency representation also presents a significant challenge for semantic communication systems, especially as modern research increasingly emphasizes reconstruction-oriented tasks. High-fidelity reconstruction demands both global structural information and local detail preservation, which cannot be achieved solely through low-frequency representations. To overcome this limitation, we propose a frequency decomposition strategy that separates semantic representations into high-and-low frequency components.

Building upon this motivation, we refine the design of the semantic encoder-decoder by integrating High-and-Low Frequency Decomposition Blocks (HiLo Blocks) as replacements for the standard transformer blocks~\cite{transformer} in ViT-based semantic communication systems~\cite{vit_semantic_1,vit_semantic_2}. These HiLo Blocks are specifically designed to capture both high-and-low frequency information through dedicated branches for each frequency component. Moreover, as demonstrated in Section~\ref{analysis}, incorporating high-low frequency decomposition further strengthens the noise-resilience of the FSQ module, improving overall robustness and performance.

In summary, our main contributions can be summarized as follows:
\begin{itemize}
    \item We conducted an in-depth analysis of the impact of semantic noise on semantic representations and illustrated the results through intuitive visualizations.
    \item To enable unified handling of multi-level noise, we propose the FSQ module to regulate semantic representations, enhancing robustness without relying on adversarial training or artificially injected noise.
    \item To further enhance the noise robustness of FSQ and mitigate the performance degradation caused by representation space constraints, we introduce a transformer-based high-and-low frequency decomposition module to decouple image representations in the frequency domain.
\end{itemize}


\section{Related Work}

\subsection{Deep Learning enabled Semantic Communication}
In recent years, numerous semantic communication systems have been developed based on advanced deep learning models. For text semantic transmission, a deep learning-based semantic communication framework, DeepSC, was proposed in~\cite{semantic_sys_1}. Compared to traditional communication paradigms, it enables transmission in low signal-to-noise ratio (SNR) environments. In the field of image transmission, Huang et al.~\cite{semantic_sys_2}. proposed a coarse-to-fine semantic communication system for image transmission based on Generative Adversarial Networks (GANs). For speech transmission, Weng et al.~\cite{sc_enc_dec_1} extended DeepSC to construct DeepSC-ST, a framework specifically designed for speech transmission. With the emergence of the Transformer~\cite{transformer} and and its state-of-the-art performance across multiple domains, recent studies~\cite{vit_semantic_1,vit_semantic_2,sc_enc_dec_2} have increasingly adopted Transformer-based architectures to design semantic communication systems.

However, these studies have primarily focused on the architecture design of semantic communication systems,  paying limited attention to the effects of semantic noise.

\subsection{Noise-Robust Semantic Communication}
With the advancement and refinement of semantic communication system frameworks, recent years have witnessed a surge of interest in addressing the challenges posed by semantic noise. Hu et al.~\cite{semantic_noise_1,semantic_noise_2} introduced a masked training strategy into semantic communication systems. Notably, in~\cite{semantic_noise_1}, they proposed a feature selection module designed to suppress noise-sensitive and task-irrelevant features, significantly enhancing system robustness and performance. With the emergence of diffusion models~\cite{diffusion}, numerous studies have investigated their application in semantic denoising. Wu et al.~\cite{semantic_noise_3,semantic_noise_4} and Xu et al. developed channel denoising modules based on diffusion models, incorporating noise during training to enable the model to progressively learn the noise distribution. Zhang et al.~\cite{semantic_noise_5} considered channel state information (CSI) and noise variance in their model design, enabling adaptive resource allocation for data transmission. 

Although these studies have explored various noise-robust strategies for semantic communication, real-world semantic noise is inherently complex and unpredictable. Existing approaches rely on adversarial training with artificially injected noise, which limits their adaptability to diverse and unseen noise patterns in practical scenarios. This underscores the need for a more flexible and generalizable approach that circumvents adversarial training and effectively mitigates complex semantic noise.

\subsection{Frequency-Domain Approaches in Neural Network}
Analyzing and addressing problems from a frequency-domain perspective is a widely adopted approach. Chen et al.~\cite{dnn_frequency_1} proposed Octave Convolution for high-and low-frequency decomposition in images. Similarly, Pan et al.~\cite{dnn_frequency_2} enhanced the Vision Transformer by introducing a High-and Low Frequency module. Leveraging the frequency-domain perspective has also enabled various applications. For instance, Kong et al.~\cite{dnn_frequency_3} developed a frequency-domain-based image deblurring algorithm, while Yu et al.~\cite{dnn_frequency_4} proposed the Frequency-Aware Spatiotemporal Transformer for video inpainting detection.

Notably, recent studies in semantic communication have also started to differentiate between various feature types. For instance, Liu et al.~\cite{semantic_frequency_1} proposed a feature importance evaluation method that dynamically allocates channel resources based on the significance of individual features.

\section{Noise-Resilient Semantic Communication with High-and-Low Frequency Decomposition}
As illustrated in Fig.~\ref{hilo_system}, this section provides the detailed introduction of our proposed Se-HiLo and  its noise robustness analysis.
\begin{figure*}[htbp]
    \centering
    \noindent\includegraphics[width=6in]{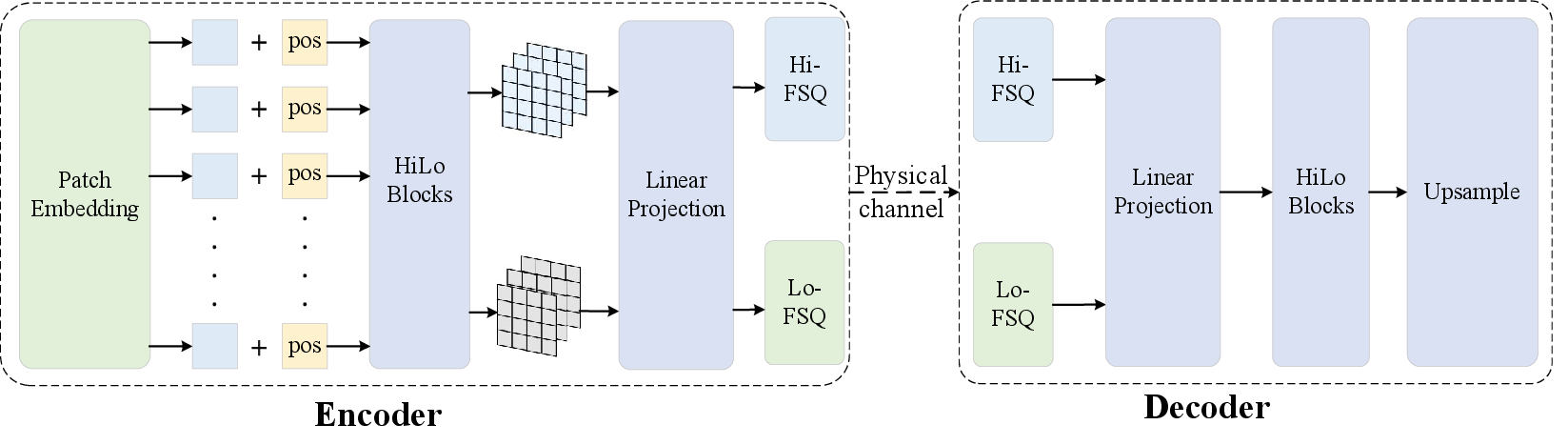}
    \caption{Noise-Resilient Semantic Communication Framework with High-and-Low-Frequency Decomposition}
    \label{hilo_system}
\end{figure*}

\subsection{Finite Scalar Quantization}

For clarity, we illustrate FSQ using single-dimensional data. Let $L$ denote the number of quantization levels. For odd $L$, levels are distributed symmetrically as $\{-k, \dots, 0, \dots, k\}$, where $k = \frac{L-1}{2}$, while for even $L$, an offset $o$ ensures alignment. The input $z \in \mathbb{R}$ is processed as follows:

\subsubsection{Bounding the Input}
To constrain $z$ within the quantization range, the bounding range is defined as 

\begin{equation}
\begin{split}
h = \frac{(L-1)(1+\epsilon)}{2},
\end{split}
\end{equation}
where $\epsilon > 0$ ensures stability. For even $L$, the offset is $o = 0.5$, while for odd $L$, $o = 0$. A shift $s = \tanh^{-1}\left(\frac{o}{h}\right)$ is applied to center the levels. The input is then transformed as $z_{\text{bounded}} = \tanh(z + s) \cdot h - o$, ensuring $z_{\text{bounded}} \in [-h, h]$.

\subsubsection{Scaling the Quantization Span}
To introduce flexibility in the quantization range, we scale the quantization span $[-h, h]$ to a new range $[-\alpha h, \alpha h]$, and $z_{\text{bounded}} = \alpha z_{\text{bounded}}$, where $\alpha > 0$ is a scaling factor.

\subsubsection{Quantization and Normalization}
The bounded value $z_{\text{bounded}}$ is quantized by rounding to the nearest level, $z_{\text{round}} = \textit{round}(\frac{z_{\text{bounded}}}{\alpha}) * \alpha$. The quantized result is normalized to $[-1, 1]$ via $z_{\text{normalized}} = z_{\text{round}} / \alpha h$ before being processed through the linear layer of the decoder.

\subsubsection{Noise Resilience and Reconstruction}
As illustrated in Fig.~\ref{fsq}, The \textit{Quantization Span} (left) defines the bounded operational region where discrete reference vectors (white circles) are positioned. These reference vectors correspond to the quantization levels, ensuring that input data is mapped to a finite set of values within this range after bounding.

The \textit{Noise Resilience Range} (right) represents the spherical region surrounding each reference vector. When a vector is perturbed by noise (red point), FSQ maps the disturbed vector directly back to its reference vector. This operation guarantees that noise-perturbed values are correctly aligned to their corresponding quantization levels without requiring explicit nearest-neighbor search.
\begin{figure}[htbp]
    \centering
    \noindent\includegraphics[width=3.2in]{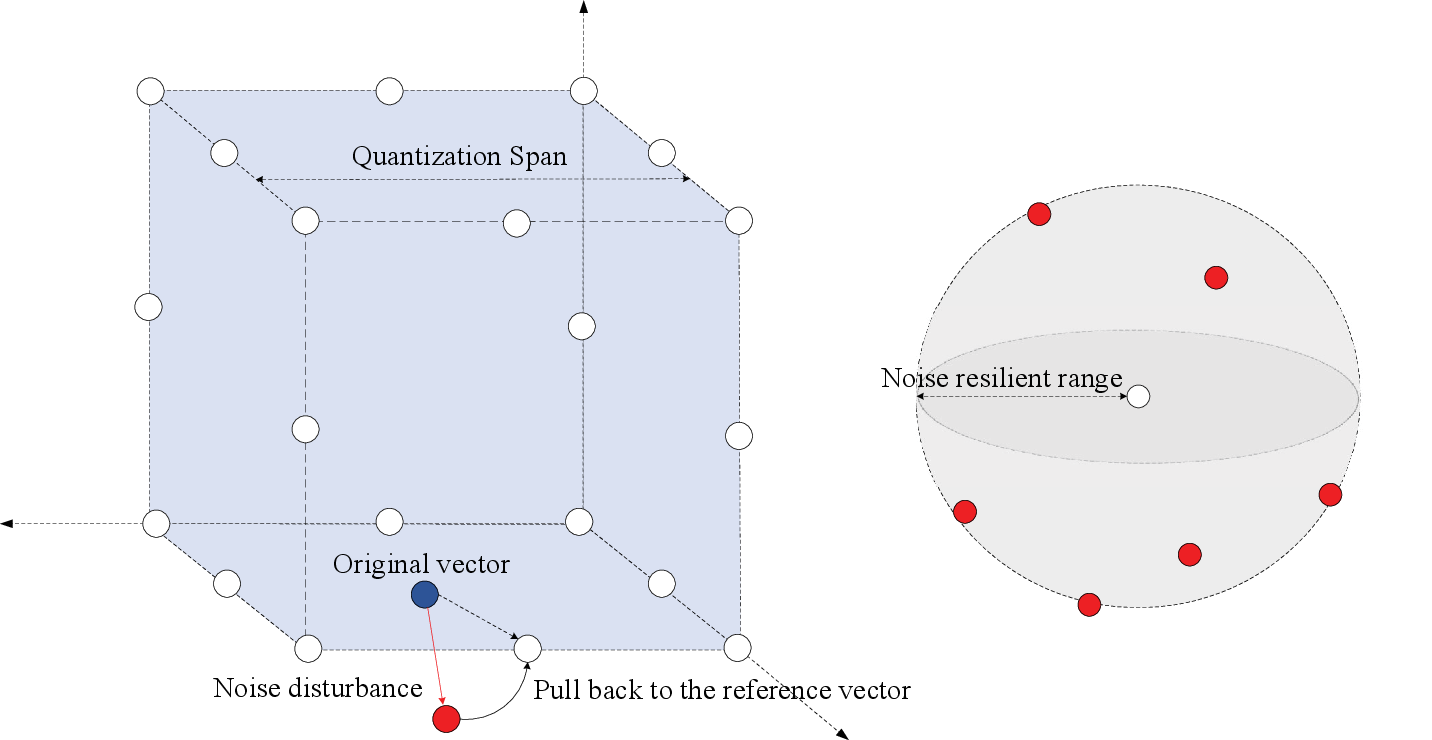}
    \caption{Finite Scalar Quantization}
    \label{fsq}
\end{figure}

\subsection{High-and Low-Frequency Decomposition}
Based on the design in \cite{dnn_frequency_2}, we propose a semantic encoder-decoder framework leveraging high-and-low frequency decomposition.
\subsubsection{High-and Low-Frequency Decomposition Block}

The HiLo Block serves as the core module of the Se-HiLo, as shown in Fig.~\ref{hiloblock}. It is specifically designed to decompose input features into high-and low-frequency components. Given an input tensor $X \in \mathbb{R}^{H \times W \times D}$, it is partitioned into high-frequency features $X_{\text{Hi}} \in \mathbb{R}^{H \times W \times D_{\text{Hi}}}$ and low-frequency features $X_{\text{Lo}} \in \mathbb{R}^{H \times W \times D_{\text{Lo}}}$, where $D_{\text{Hi}} + D_{\text{Lo}} = D$. 

Low-frequency features are downsampled to $X_{\text{Lo\_down}} \in \mathbb{R}^{\frac{H}{s} \times \frac{W}{s} \times D_{\text{Lo}}}$ and processed through a Lo-ViT Block utilizing cross-attention mechanisms to capture global contextual information. Meanwhile, high-frequency features are spatially divided into smaller groups, where self-attention mechanisms in a Hi-ViT Block are applied to extract fine-grained local details.

The outputs from both pathways are concatenated to generate the final output $X_{\text{Out}} \in \mathbb{R}^{H \times W \times D}$. This design effectively balances the preservation of local details and the representation of global contextual information.
\begin{figure}[htbp]
    \centering
    \noindent\includegraphics[width=3.4in]{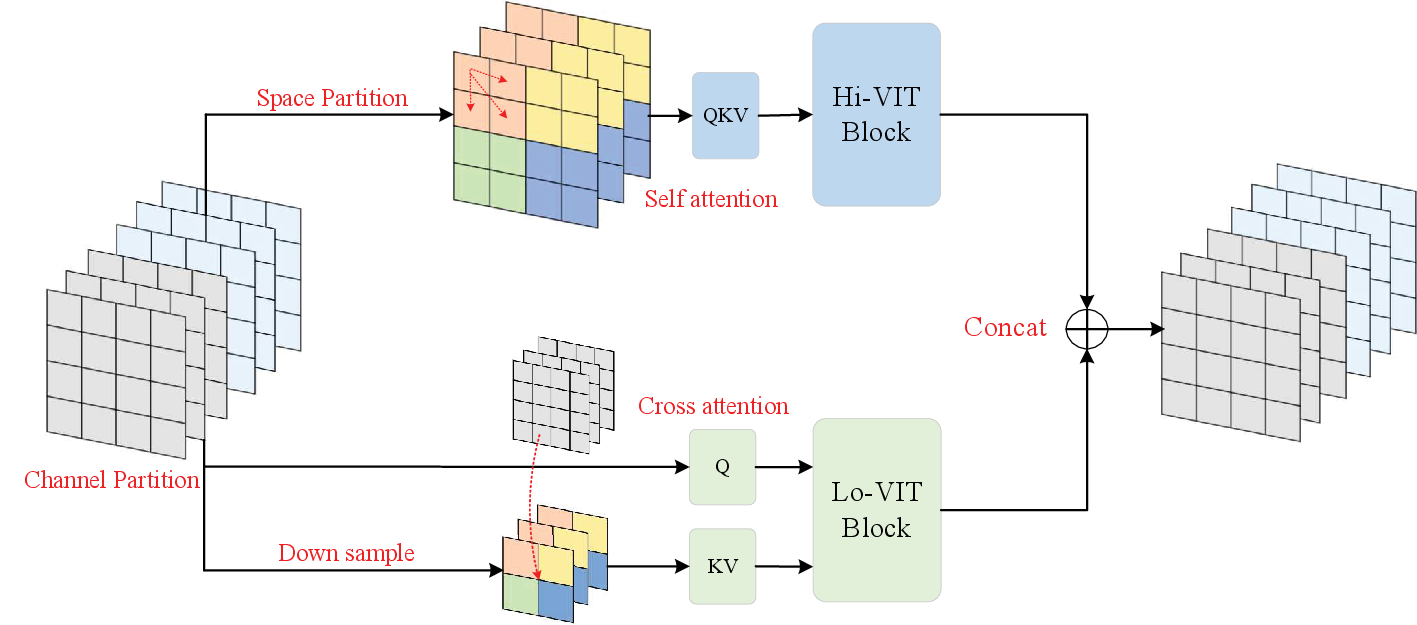}
    \caption{High-and-Low-Frequency Decomposition Block}
    \label{hiloblock}
\end{figure}
\subsubsection{Encoder Design}

The encoder consists of multiple HiLo blocks, which transforms the input data $I \in \mathbb{R}^{H \times W \times C}$ into high-frequency and low-frequency representations. Initially, the input data is divided into patches and embedded, producing $P \in \mathbb{R}^{N \times D}$, where $N$ represents the number of patches and $D$ denotes the representation dimension. These embedded patches are processed through multiple HiLo Blocks, decomposing the representations into high-frequency representations $P_{\text{Hi}} \in \mathbb{R}^{N_{\text{Hi}} \times D_{\text{Hi}}}$ and low-frequency representations $P_{\text{Lo}} \in \mathbb{R}^{N_{\text{Lo}} \times D_{\text{Lo}}}$. 

The high-and-low frequency representations are then mapped through separate linear transformations to align with the FSQ dimensions, producing $T_{\text{Hi}} \in \mathbb{R}^{{N} \times D_{\text{FSQ}}}$ and $T_{\text{Lo}} \in \mathbb{R}^{{N} \times D_{\text{FSQ}}}$, where $D_{\text{FSQ}}$ represents the dimensionality required for FSQ processing. Finally, the high-frequency and low-frequency representations are transmitted through physical channels.

\subsubsection{Decoder Design}

The decoder reconstructs the transmitted high-and-low frequency representations into the original image. The received high-frequency representation $R_{\text{Hi}} \in \mathbb{R}^{N \times D_{\text{FSQ}}}$ and low-frequency representation $R_{\text{Lo}} \in \mathbb{R}^{N \times D_{\text{FSQ}}}$ are first processed through frequency-specific reconstruction modules (FSQ). Following FSQ processing, each representation is linearly mapped back to its original feature dimensions, yielding $P_{\text{Hi\_rec}} \in \mathbb{R}^{N \times D_{\text{Hi}}}$ and $P_{\text{Lo\_rec}} \in \mathbb{R}^{N \times D_{\text{Lo}}}$. 

The reconstructed high-and-low frequency representations are then concatenated along the channel dimension to form $z \in \mathbb{R}^{N \times D}$. Similarly, 
$z$ is processed through multiple HiLo Blocks in the decoder. Finally, the refined representations are upsampled to the original spatial resolution, reconstructing the image $\hat{I} \in \mathbb{R}^{H \times W \times C}$, which closely approximates the original input $I$.

\subsubsection{Loss Function}
The loss function in this paper consists of four components. The first component is the reconstruction loss:
\begin{equation}
\begin{split}
\mathcal{L}_{\text{recon}} = \| \hat{I} - I \|_2^2,
\end{split}
\end{equation}
the second component is the perceptual loss:
\begin{equation}
\begin{split}
\mathcal{L}_{\text{perceptual}} = \| \phi(\hat{I}) - \phi(I) \|_2^2,
\end{split}
\end{equation}
where \( \phi \) represents the feature maps of a pre-trained network, such as VGG. The third component is the discriminator loss:
\begin{equation}
\begin{split}
\mathcal{L}_{\text{adv}} = -\mathbb{E}_{x \sim p_{\text{data}}}[\log D(I)] - \mathbb{E}_{z \sim p_{\text{z}}}[\log(1 - D(G(z)))],
\end{split}
\end{equation}
where \( D \) is the discriminator, \( G \) is the generator. 
The final loss function is a weighted sum of these components:
\begin{equation}
\begin{split}
\mathcal{L} = \mathcal{L}_{\text{recon}} + \lambda_1 \mathcal{L}_{\text{perceptual}}  + \lambda_2 \mathcal{L}_{\text{adv}}.
\end{split}
\end{equation}

\subsection{Noise Robustness Analysis} \label{analysis}
\subsubsection{Quantization}
Given an $N$-dimensional representation vector $\mathbf{x} = [x_1, x_2, \dots, x_N]$, it is quantized into a finite set of discrete values to constrain its representation. Assume each dimension spans the interval $[L, U]$, and the quantization level is specified by $\text{level} = [m_1, m_2, \dots, m_N]$. For the $i$-th dimension, the quantization points are:
\begin{equation}
    v_{ij} = L + j \cdot \Delta, \quad j = 0, 1, \dots, m_i - 1,
\end{equation}
where the step size $\Delta$ is:
\begin{equation}
    \Delta = \frac{U - L}{m_i - 1}.
\end{equation}
Each component $x_i$ is quantized to the closest point:
\begin{equation}
    \hat{x}_i = \arg\min_{v_{ij} \in \{L, L + \Delta, \dots, U\}} |x_i - v_{ij}|.
\end{equation}

\subsubsection{Single-Dimension Noise Robustness}
\begin{theorem}
For a single-dimension input $x_i$ quantized into $m_i$ levels over $[L, U]$ with step size $\Delta$ under additive white Gaussian noise (AWGN) $\eta_i \sim \mathcal{N}(0, \sigma^2)$, the probability of correct quantization is:
\begin{equation}
    P_{\text{correct}, i} = \text{erf}\left(\frac{\Delta}{2\sqrt{2}\sigma}\right).
\end{equation}
\end{theorem}

\begin{proof}
Given the received signal:
\begin{equation}
    y_i = x_i + \eta_i,
\end{equation}
quantization is correct if:
\begin{equation}
    |\eta_i| \leq \frac{\Delta}{2}.
\end{equation}
The probability is:
\begin{equation}
    P_{\text{correct}, i} = \int_{-\Delta/2}^{\Delta/2} \frac{1}{\sqrt{2\pi}\sigma} e^{-\frac{\eta_i^2}{2\sigma^2}} d\eta_i.
\end{equation}
Substituting:
\begin{equation}
    z = \frac{\eta_i}{\sqrt{2}\sigma}, \quad d\eta_i = \sqrt{2}\sigma \, dz,
\end{equation}
transforms the integral:
\begin{equation}
    P_{\text{correct}, i} = \frac{1}{\sqrt{\pi}} \int_{-\frac{\Delta}{2\sqrt{2}\sigma}}^{\frac{\Delta}{2\sqrt{2}\sigma}} e^{-z^2} dz.
\end{equation}
Using the error function:
\begin{equation}
    \text{erf}(x) = \frac{2}{\sqrt{\pi}} \int_0^x e^{-z^2} dz,
\end{equation}
we simplify to:
\begin{equation}
    P_{\text{correct}, i} = \text{erf}\left(\frac{\Delta}{2\sqrt{2}\sigma}\right).
\end{equation}
\end{proof}

\subsubsection{Multi-Dimensional Noise Robustness}

\begin{theorem}
For an $N$-dimensional vector $\mathbf{x}$ with independent Gaussian noise $\eta_i \sim \mathcal{N}(0, \sigma^2)$ in each dimension, the overall probability of correct quantization is:
\begin{equation}
    P_{\text{correct}} = \left[\text{erf}\left(\frac{\Delta}{2\sqrt{2}\sigma}\right)\right]^N.
\end{equation}
\end{theorem}
\begin{proof}
By leveraging the conclusions from the single-dimensional case, we can easily extend them to derive the noise robustness in the multi-dimensional scenario.
\end{proof}

Based on the analysis above, we can draw several conclusions:
\begin{itemize}
\item \textbf{Quantization Levels ($m_i$)}: Higher $m_i$ improves precision but reduces noise tolerance; lower $m_i$ enhances noise robustness but sacrifices resolution.
    \item \textbf{Quantization Span ($U - L$)}: Larger spans increase noise robustness but risk losing fine details if inputs are concentrated in narrow ranges.
    \item \textbf{Noise Variance ($\sigma^2$)}: Lower variance improves robustness, while higher variance reduces it.
    \item \textbf{Dimensionality ($N$)}: Robustness decreases exponentially with $N$ as errors compound across dimensions.
\end{itemize}

Since noise cannot be directly controlled, the proposed Se-HiLo reduces the dimensionality of transmitted features \(N\) by decomposing them into high-and-low frequency componentss. Additionally, during FSQ quantization, the high-and-low frequency representations are processed in independent quantization spaces. This design allows for smaller quantization levels while inherently expanding the quantization span, thereby enhancing robustness and transmission efficiency in noisy environments.

\section{Experiments}
\subsection{Implement Details}
The model was trained on the \href{https://www.kaggle.com/datasets/alessiocorrado99/animals10}{Animal-10 Dataset} (it contains about 28K medium quality animal images belonging to 10 categories) with a batch size of 32 for 50 epochs using the Adam optimizer and a learning rate of 0.0001. The dropout rate was set to 0.1. Quantization levels were configured as [5, 5, 5, 5, 5] with a scaling factor $\alpha = 2$. Both high-and low-frequency channels had 256 dimensions. The encoder and decoder each contained 8 HiLo blocks.

\subsection{Evaluation Metrics}
PSNR evaluates image quality based on pixel-wise differences, with higher values indicating better reconstruction. SSIM measures perceptual similarity by considering structural information, where values closer to 1 represent higher similarity.

\subsection{Baseline}
We compare our method with TiTok~\cite{TiTok}, VQGAN~\cite{vqgan}, ViT-VQGAN~\cite{vitvqgan}, and MoVQGAN~\cite{movq}. Titok demonstrates the capability to compress images into 32 tokens, achieving an impressive compression ratio. VQGAN leverages vector quantization and adversarial training to improve image generation quality. ViT-VQGAN incorporates Vision Transformers (ViT) to capture long-range dependencies, enhancing feature representation. MoVQGAN proposes multichannel quantization to improve codebook utilization and reconstruction performance.
\subsection{Results}
\subsubsection{Comparisons with Baselines under Different SNR}
As illustrated in Fig.~\ref{exp_1}, as the SNR decreases, all methods exhibit varying degrees of quality degradation due to increased noise interference. Titok is highly susceptible to noise interference, and even under a 10dB SNR, the reconstructed image categories exhibit significant errors. VQGAN performs better, preserving structural information at higher SNRs, but produces noticeable blurring and distortions as noise increases. MoVQGAN and ViT-VQGAN achieve acceptable results under moderate noise conditions but fails to handle extreme noise, leading to severe artifacts at low SNRs. 
\begin{figure}[htbp]
    \centering
    \hspace*{-1cm}
    \includegraphics[width=3.4in]{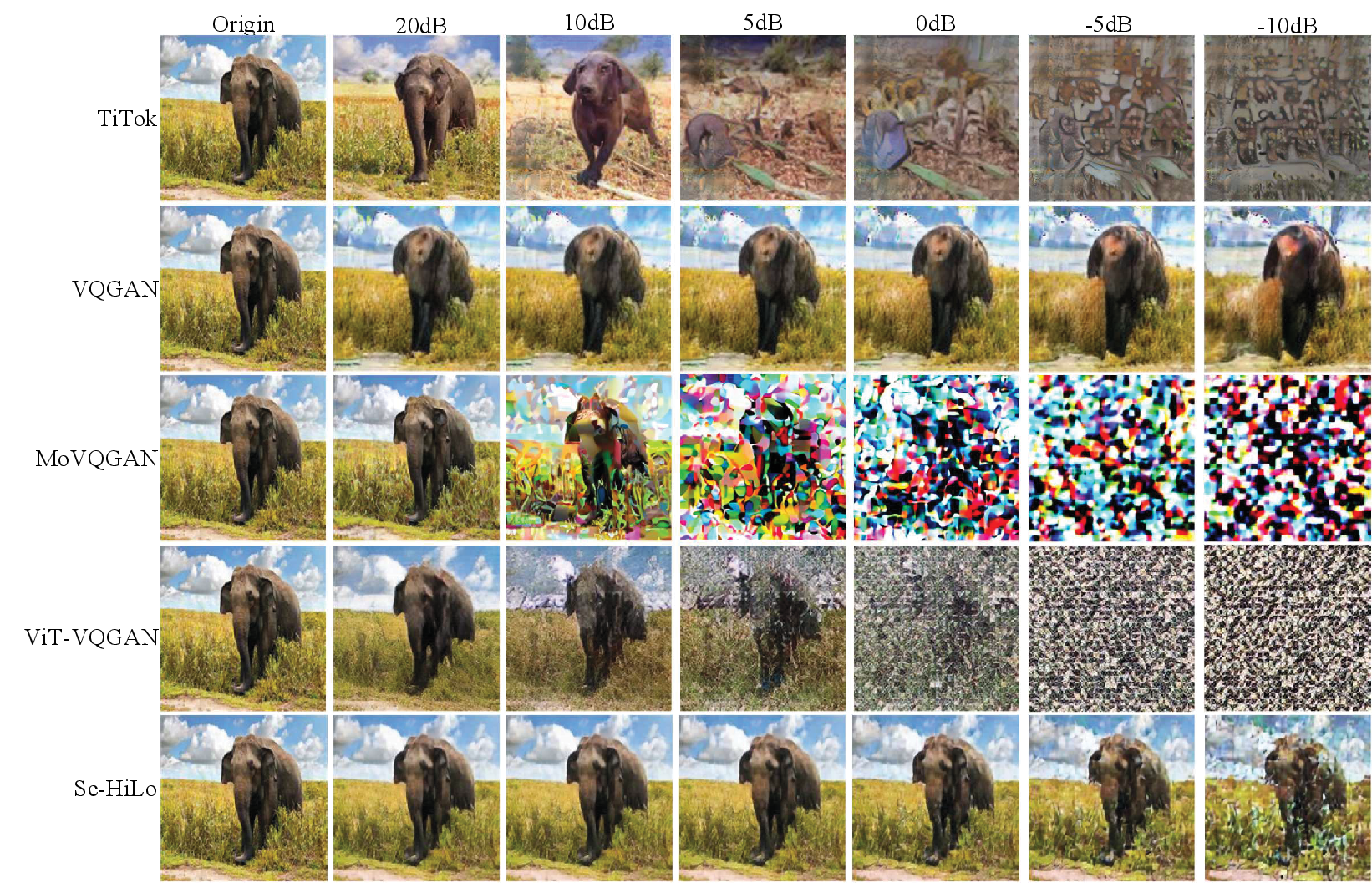}
    \caption{Visual Comparisons under Different SNR.}
    \label{exp_1}
\end{figure}

In contrast, the proposed Se-HiLo demonstrates superior performance across all SNR levels. It effectively preserves structural details and suppresses noise artifacts, maintaining high visual fidelity even under challenging conditions. These results underscore the robustness and adaptability of Se-HiLo in low-SNR environments.

To complement the visual comparisons, Fig.~\ref{exp_2} and Fig.~\ref{exp_3} present quantitative evaluations using PSNR and SSIM metrics under different SNR.
\begin{figure}[htbp]
\centering 
\includegraphics[width=2.3in]{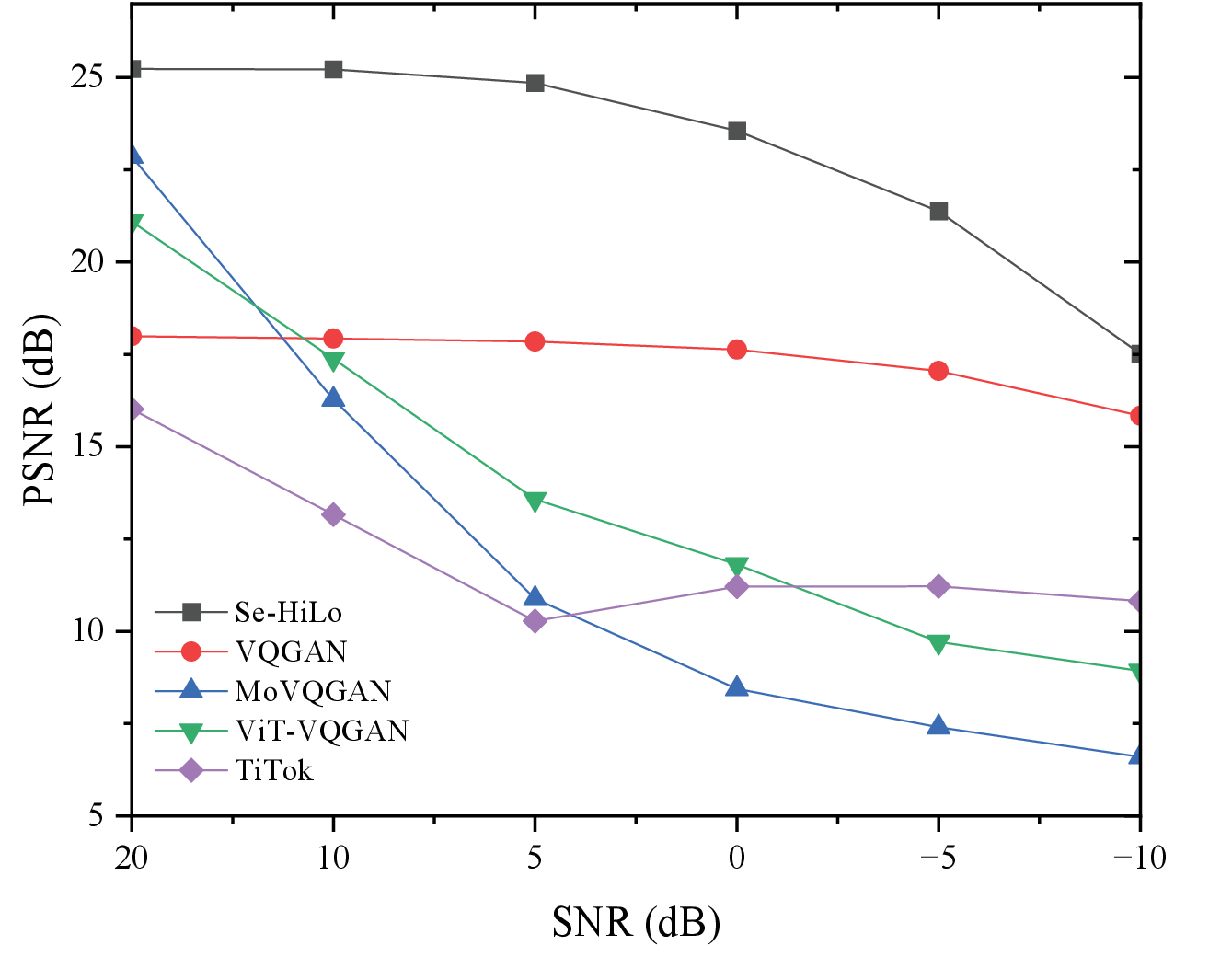} 
\caption{PSNR Performance under Different SNR.} 
\label{exp_2}
\end{figure}
    
\begin{figure}[htbp]
\centering
\includegraphics[width=2.3in]{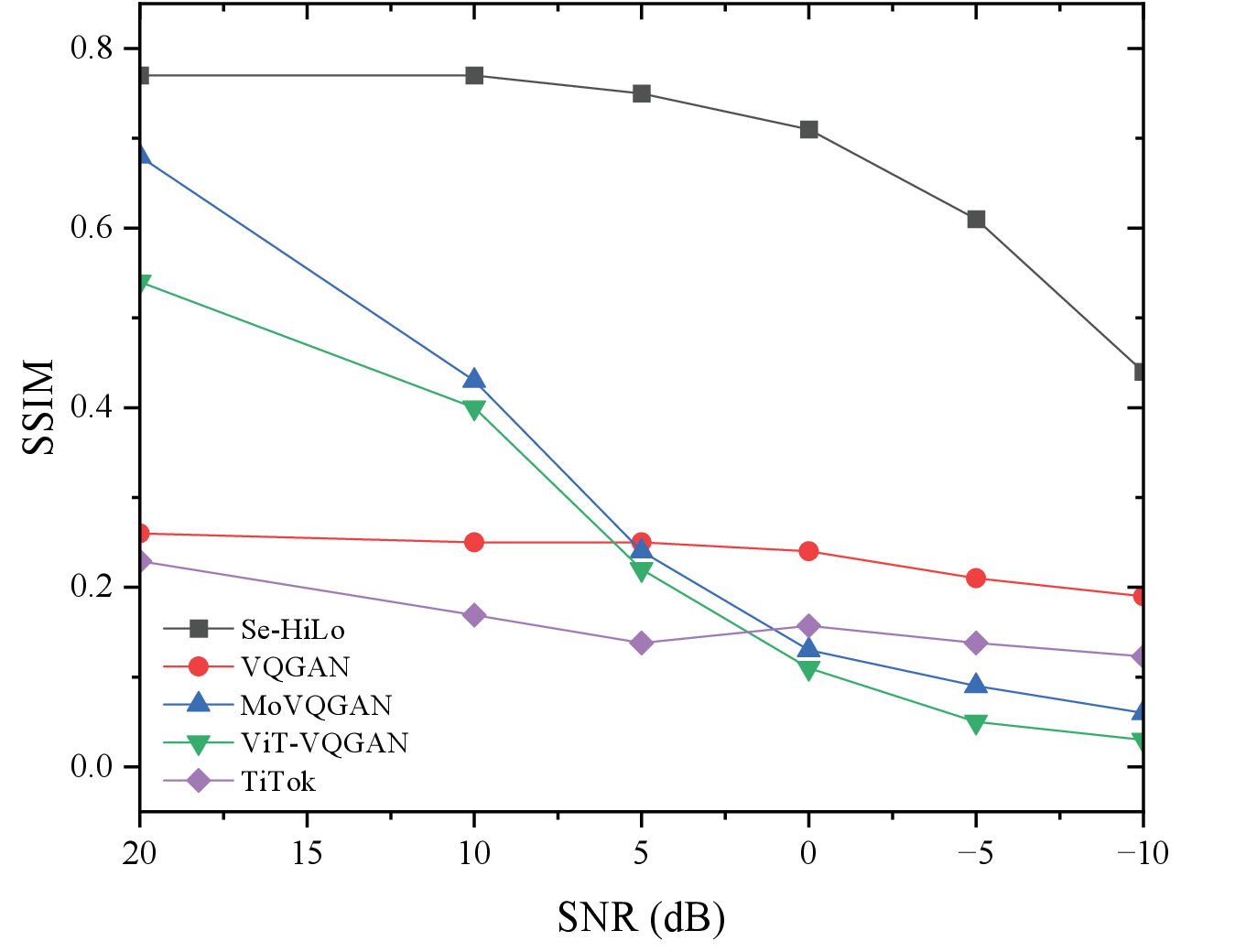}
\caption{SSIM Performance under Different SNR.} 
\label{exp_3}
\end{figure}

The proposed Se-HiLo achieves consistently higher PSNR across all SNR levels, demonstrating superior reconstruction quality. As SNR decreases, the performance of all models declines; however, Se-HiLo exhibits greater robustness, maintaining higher PSNR values even under severe noise conditions. At lower SNRs, the SSIM of MoVQGAN and ViT-VQGAN drop sharply, indicating poor structural preservation. In contrast, Se-HiLo demonstrates better resilience, retaining higher SSIM scores even in noisy environments.

To further demonstrate the superiority of Se-HiLo, we introduced artificial noise into the VQGAN, which had shown favorable performance in previous experiments, and subjected it to adversarial training. As shown in Fig.~\ref{loss}, we observed that, in order to adapt to higher noise levels, the model exhibited unstable loss curves and even failed to converge. This observation highlights that the noise resistance capability of neural networks has an upper bound. In contrast, Se-HiLo circumvents adversarial training and achieves noise robustness by simply adjusting the scaling factor $\alpha$.

\begin{figure}[htbp]
\centering
\includegraphics[width=2.3in]{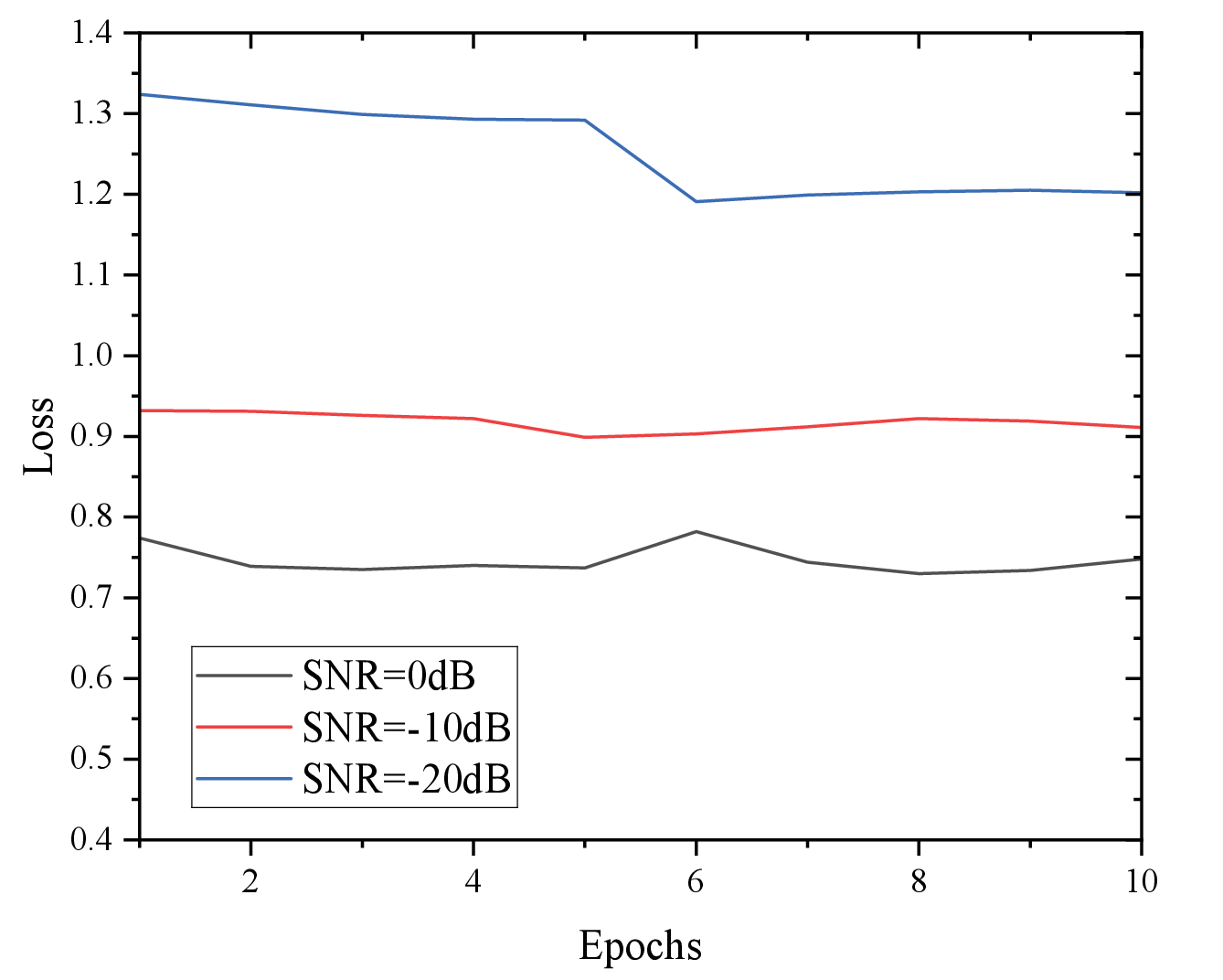}
\caption{Loss Curve of VQGAN with Adversarial Training.} 
\label{loss}
\end{figure}

\subsubsection{High-and-Low Frequency Representations under Different SNR}

We investigated the impact of noise on representations at different frequencies, as shown in Fig~\ref{exp_4}. When fixing the noise level of high-frequency representations and gradually increasing the noise in low-frequency representations, the global structure of the image becomes progressively blurred. However, the color and texture of the image remain relatively well-preserved. In contrast, when gradually increasing the noise in high-frequency representations, severe distortions in color and texture are observed, while the overall outline of the image remains relatively intact.

\begin{figure}[htbp]
    \centering
    \includegraphics[width=3in]{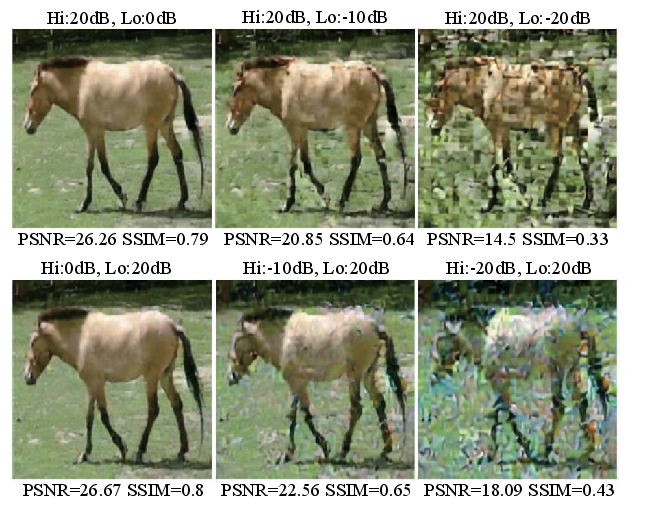}
    \caption{High-and-Low Frequency Representations under Different SNR}
    \label{exp_4}
\end{figure}
\subsubsection{Impact of High-and-Low Frequency Proportions}
We further adjusted the proportion of high- and low-frequency components during training. As shown in Fig~\ref{exp_5}, the experiments demonstrate that performance is optimal when the proportion of high-and-low frequency components is balanced. However, when the frequency component proportion becomes imbalanced, the performance degrades significantly under strong semantic noise interference.

\begin{figure}[htbp]
    \centering
    \includegraphics[width=2.8in]{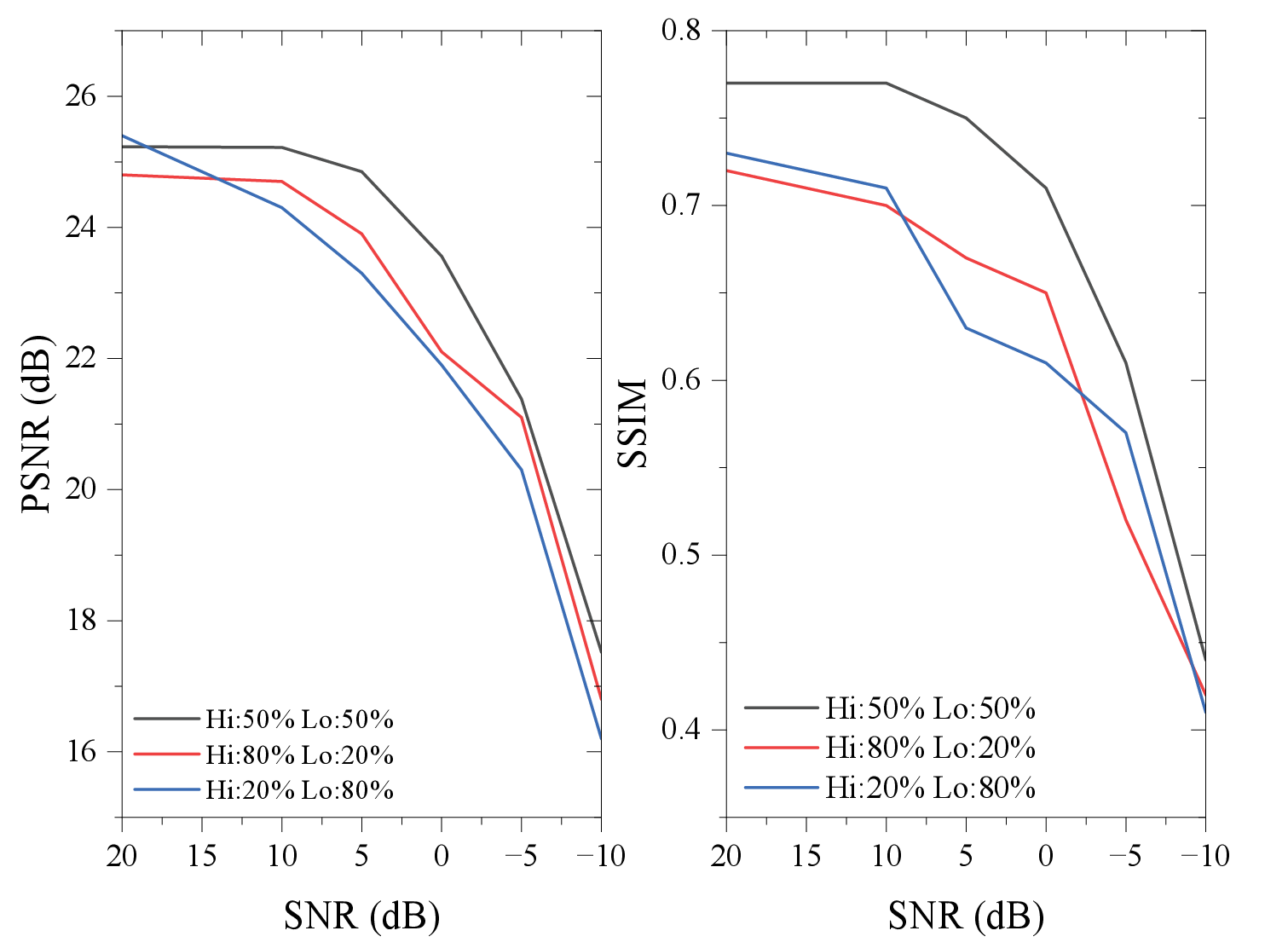}
    \caption{Impact of High-and-Low Frequency Proportions.}
    \label{exp_5}
\end{figure}

\section{Conclusion}
This paper investigates the susceptibility of representation vectors to semantic noise in semantic communication. To address this issue, we propose a Finite Scalar Quantization (FSQ) module that mitigates noise interference without adversarial training, providing a computationally efficient and scalable solution. To compensate for potential information loss caused by FSQ, we further introduce a transformer-based high-and-low frequency decomposition module. By separating representations into distinct frequency components, our approach preserves both global structures and fine-grained details, enhancing noise robustness.

\newpage
\bibliographystyle{named}
\bibliography{ijcai25}

\end{document}